\documentclass[10pt]{amsart}
\usepackage[english]{babel}
\usepackage{natbib}
\usepackage{amsmath,amssymb,amsthm,bm}
\usepackage{natbib}

\usepackage{paralist}

\newcommand{\C}{\mathbb{C}}
\newcommand{\ii}{i}
\newcommand{\intd}{\mspace{0.5mu}\operatorname{d}\mspace{-2.5mu}}
\newcommand{\bigoh}{\mathcal{O}}
\newcommand{\RE}{\operatorname{Re}}
\newcommand{\IM}{\operatorname{Im}}
\newcommand{\Div}{\operatorname{div}}
\newcommand{\sub}{\operatorname{sub}}
\newcommand{\range}{\operatorname{range}}
\newcommand{\pd}{pseudo-differential}

\newcommand{\strain}{\epsilon}
\newcommand{\stress}{\sigma}
\newcommand{\traction}{\tau}
\newcommand{\wavenumber}{\kappa}
\newcommand{\frequency}{\omega}
\newcommand{\asymptotic}{\sim}

\newcommand{\zitat}[1]{\citep{#1}}
\newcommand{\zitatthm}[2]{\citep[#2]{#1}}
\newcommand{\zitatau}[2]{\citet{#1}}
\newcommand{\zitatauthm}[3]{\citet[#3]{#1}}

\theoremstyle{plain}
\newtheorem{theorem}{Theorem}[section]
\newtheorem{proposition}{Proposition}[section]
\newtheorem{lemma}{Lemma}[section]
\theoremstyle{remark}
\newtheorem{remark}{Remark}[section]
\newtheorem*{ack}{Acknowledgement}

\usepackage{hyperref}
\renewcommand{\harvardurl}[1]{\textbf{URL:} \url{#1}}

\title{Subsonic Free Surface Waves in Linear Elasticity}

\author{S\"onke Hansen}
\address{Institut f\"ur Mathematik, Universit\"at Paderborn, 33095 Paderborn, Germany}
\date{\today}
\email{soenke@math.upb.de}

\begin{document}

\begin{abstract}
For general anisotropic linear elastic solids with smooth boundaries,
Rayleigh-type surface waves are studied.
Using spectral factorizations of matrix polynomials, a self-contained exposition 
of the case of a homogeneous half-space is given first.
The main result is about inhomogeneous anisotropic bodies with curved surfaces.
The existence of subsonic free surface waves is shown by giving
ray series asymptotic expansions, including formulas for the transport equation.
\end{abstract}

\maketitle

\section{Introduction}

\zitatau{Rayleigh1887surface}{Rayleigh} discovered waves which propagate
along a plane traction-free surface of an isotropic and homogeneous elastic solid.
The surface wave speeds are subsonic, that is,
they are strictly less than the wave speeds of interior body waves.
Furthermore, the amplitudes attenuate exponentially with distance to the surface.
\zitatau{Synge57elast}{Synge} raised doubts about existence of Rayleigh-type
waves in anisotropic solids.
\zitatau{Stroh62}{Stroh} pointed out that these doubts were unfounded, and he
introduced a sextic eigenvalue problem which became useful
in the theory of free surface waves in anisotropic solids.
In the early 1970's, the existence and uniqueness problem of free surface homogeneous
plane waves in a semi-infinite half-space was settled by Barnett, Lothe, and coworkers.
For any given horizontal propagation direction they showed that there is at most one
free surface wave speed, and they gave criteria for the existence of such waves.
\zitatau{LotheBarnett76extsurfwave}{Lothe and Barnett} rederived their results by the surface impedance tensor method.
The surface impedance tensor
relates the surface displacement to the surface traction required to sustain it. 
This tensor was introduced by \zitatau{IngebTonning69}{Ingebrigtsen and Tonning}.
Detailed presentations of the existence and uniqueness results were
given by \zitatau{ChadwickSmith77}{Chadwick and Smith} and \zitatau{LotheBarnett85surfwaveimped}{Barnett and Lothe}.
Much later,
\zitatau{MielkeFu04surfwavespeed}{Mielke and Fu} simplified some proofs of the Barnett--Lothe
theory by using a Ricatti equation satisfied by the impedance tensor.
A crucial property of the tensor, the positive definiteness of its real part,
follows from an integral identity which, in the original treatments,
arises somewhat magically by averaging over rotations
in the plane spanned by the normal to the surface and the propagation direction.
Existence of subsonic free surface waves was shown by \zitatau{KamKiselev09Rayleigh}{Kamotski{\u\i} and Kiselev}
with a completely different approach based on the variational principle.

Concerning Rayleigh-type waves in inhomogeneous elastic solids with curved boundaries,
\zitatau{Petrowsky45}{Petrowsky} exhibited the following locality principle:
If a surface wave exists, the velocity of its discontinuity
at a given point must be equal to the velocity in the homogeneous elastic half-space
which is obtained by freezing the elastic parameters at that point.
This locality principle is efficiently implemented by asymptotic ray methods,
which substitute, in the high-frequency regime,
standard plane waves by geometrico-optical `plane waves'.
For inhomogeneous, isotropic elastic solids with curved boundaries,
ray methods involving sums of complex plane waves were successfully applied to the free
surface wave problem by \zitatau{Babich61Rayleigh}{Babich}, \zitatau{BabichRusakova62}{Babich and Rusakova},
and by \zitatau{KaralKeller64elasurfwave}{Karal and Keller}.
The ray method works, for smooth data, under assumptions on the geometry and on the
elastic parameters which are less restrictive than those needed for finding analytic solutions.
More importantly, salient features of high-frequency waves such as wave fronts and
amplitudes are captured directly by ray methods.
Using the existence and regularity theory of linear hyperbolic equations, as was
done by \zitatau{CourantLax56propagation}{Courant and Lax}, one can correct asymptotic solutions,
without changing the important high-frequency properties, into genuine exact solutions.
\zitatau{Gregory71Rayleigh}{Gregory} compared some analytic representations of surface waves
and corresponding ray approximations.
Rayleigh surface waves in inhomogeneous, anisotropic elastic bodies
were studied with the complex ray method by \zitatau{Nomofilov79propRayl}{Nomofilov}.
The amplitude of a geometrico-optical wave is governed by a transport equation
which, in the case of free surface waves, is quite complicated.
Efforts to solve the transport equations culminated in the work of
\zitatau{BabichKirpich04rayleigh}{Babich and Kirpichnikova}, where detailed formulas for the amplitude
and the Berry phase of a Rayleigh surface wave were obtained.

In the present paper we study Rayleigh-type surface waves in anisotropic elastic solids
with smooth surfaces and smooth inhomogeneities.
We give a self-contained presentation of the theory of Barnett and Lothe,
and we incorporate it into a ray theory.
Under the same conditions as for homogeneous half-spaces,
the existence of subsonic free surface waves is proved.
Transport equations for leading amplitudes are established in a
way which enables their numerical solution.

Rayleigh wave speeds are frequency dependent if, as happens for coated materials, 
material parameters vary significantly over one wavelength;
see \zitatau{Destrade07}{Destrade} for the study of a model case.
Our analysis does not deal with this dispersion phenomenon.
As in \zitat{BabichKirpich04rayleigh}, 
we restrict attention to the high-frequency asymptotics of Rayleigh-type surface waves 
assuming that the elastic parameters vary slowly in a (thick) surface layer.
Dispersive behaviour related to surface waves, not for Rayleigh waves but for whispering galley modes
associated with gliding rays of a scalar wave equation, has been shown quite recently
by Ivanovici, Lebeau and Planchon,
\zitat{Ivanovici12counterex} and \zitat{IvanoviciLebeauPlanchon14dispers}.

We use spectral factorizations of the acoustic tensor, which is regarded as
a second order matrix polynomial in the variable conormal to the boundary.
This allows to conveniently lump together the relevant complex eigenvalues and to
avoid some cumbersome and unnecessary considerations.
Moreover, the factorizations reduce the elastodynamic system, in the subsonic regime,
to a first order system.
In general, when the surface is curved and the solid inhomogeneous,
the first order system is not differential but \pd.
The zero traction boundary problem is transformed into a \pd\ wave equation
on the space-time boundary.
The associated principal symbol is the surface impedance tensor.
The \pd\ wave equation is only defined microlocally over the subsonic region.
Still it can be treated ray theoretically since it has the structure
of a so-called real principal type system \zitat{Dencker82polarsets}.
In the microlocal analysis literature 
such an approach was carried out for the isotropic case by \zitatau{Taylor79rayleighwaves}{Taylor},
where the subsonic region is called the elliptic region because
the theory of elliptic boundary value problems applies in the reduction to the boundary.
\zitatau{Nakamura91rayleighpulses}{Nakamura} treated the Barnett--Lothe method
from the point of view of microlocal analysis.
The present paper uses ideas developed in \zitat{SHRoehrig04lagrsol,SH11rqm,SH12imptensor}.
Although the spirit and the techniques of microlocal analysis are key to our approach,
we shall give, in order to make the contents readable for a wide audience,
a self-contained presentation except for using few basic results from \pd\ calculus.

The paper is organized as follows.
Section~\ref{section-geometrical-language} recalls the free surface traction
problem in differential geometric tensor notation.
In sections \ref{section-homog-halfspace} and \ref{sect-surface-impedance-tensor}
the subsonic free surface wave theory for homogeneous elastic half-spaces
is redeveloped using division theory of matrix polynomials as the basic tool.
In the appendix of the paper, we give a complete and self-contained presentation
of those parts of division theory which we employ.
The remaining sections deal with inhomogeneous solids with curved boundaries.
For the benefit of readers who are not familiar with \pd\ operators,
Section~\ref{sect-asymptotic-psido} contains an account of core properties
of \pd\ calculus and its relation to asymptotic expansions.
A ray theory for \pd\ systems is developed in Section~\ref{section-ray-theory}.
In Section~\ref{sect-displacement-bp} the displacement boundary problem
is solved for surface displacements concentrated in the subsonic region.
A factorization of the elastodynamic operator is constructed and used in that section.
Theorem~\ref{main-thm-exist-ssw} in Section~\ref{sect-exist-freesurfwave}
is our main result on the existence of Rayleigh-type waves.
The transport equation for the leading amplitude of subsonic free surface waves
is treated in some detail in Section~\ref{sect-transport-eqn}.

\section{Equations of linear elastodynamics}
\label{section-geometrical-language}

Let $B$ be a linearly elastic body with density $\rho>0$ and
stiffness tensor $C=[C^{ijk\ell}]$.
The elasticities satisfy the standard symmetries,
\begin{equation}
\label{elasticities-symm-def}
C^{ijk\ell}=C^{k\ell ij} =C^{jik\ell},
\end{equation}
and the strong convexity property,
\begin{equation}
\label{elasticities-strong-convex}
C^{ijk\ell}\strain_{ij}\overline{\strain_{k\ell}} >0
\quad\text{if $\strain_{ji}=\strain_{ij}$, $\strain\equiv [\strain_{ij}]\neq 0$}.
\end{equation}
(We use the summation convention, and
we denote a point, if at all, by its coordinates $x=(x^1,x^2,x^3)$.
A bar denotes complex conjugation.)
Since the elasticities are real, it suffices to assume that
\eqref{elasticities-strong-convex} holds for symmetric tensors $\strain$ which are real.
Assumptions \eqref{elasticities-symm-def} and \eqref{elasticities-strong-convex}
say that $C$ defines an inner product on symmetric $2$-tensors.

The strain tensor field $\strain$ measures the deformation of the geometry of $B$
caused by an (infinitesimal) displacement $u$.
The geometry is given by the Riemannian metric tensor $G=[G_{ij}]$.
Precisely, strain is the symmetrized covariant derivative of displacement:
\[
\strain_{k\ell}\equiv \big(u_{k;\ell} +  u_{\ell;k}\big)/2, \quad
u_{k;\ell} \equiv  u_{k,\ell} - \Gamma^j_{k\ell} u_j.
\]
Here
$\Gamma^i_{k\ell}\equiv G^{ij}(G_{\ell j,k}+G_{kj,\ell}-G_{k\ell,j})/2$
are the standard Christoffel symbols of differential geometry.
We precede a coordinate index by a comma or a semicolon to denote
a partial or a covariant derivative, respectively.
The stress field, $\stress=[\stress^{ij}]$, represents forces.
It is related to strain via Hooke's Law:
\begin{equation}
\label{eq-Hooke}
\stress^{ij}= C^{ijk\ell} \strain_{k\ell} = C^{ijk\ell} u_{k;\ell}.
\end{equation}

Isotropy is a symmetry property which some elastic bodies have.
In this case the elasticity tensor is given by
\[
C^{ijk\ell}=\lambda G^{ij}G^{k\ell}
    +\mu(G^{ik}G^{j\ell} +G^{i\ell}G^{jk}),
\]
where $\lambda$ and $\mu$ are the Lam{\'e} parameters.

The inverse of $G$ is $G^{-1}=[G^{ij}]$.
By raising and lowering indices one switches
between covariant and contravariant components, e.g.,
$u^i=G^{ik}u_k$ and $u_k=G_{kj}u^j$.
The metric $G$ defines the length element $\intd s$ by
$\intd s^2=G_{ij}\intd x^i \intd x^j$, and the volume element
$\intd V\equiv \sqrt{g}\intd x$, $g\equiv \det G$.

Let $S$ denote the (smooth) boundary surface of $B$, $\intd S$ its surface element,
and $\nu$ the interior unit normal field.
Assume that $B$ is source-free, and that $S$ is traction-free.
Then the displacement field satisfies the free surface boundary problem:
\begin{equation}
\label{eq-elastodyn-freebdry}
\rho {\ddot u}^i - \stress^{ij}_{\phantom{ij};j} = 0,
\qquad
\stress^{ij}\nu_j|_S=0,
\end{equation}
$i=1,2,3$.
The differential equations are the $3\times 3$ system of elastodynamics.
The equations~\eqref{eq-elastodyn-freebdry} arise as the Euler equations
of the Lagrangian which is the functional given by
\[
{\mathcal L}(u)
  \equiv \int\int_B (\rho {\dot u}^i\overline{{\dot u}_i}
      - \stress^{ij}\overline{\strain_{ij}})\intd V\intd t.
\]
The divergence of the stress tensor is
\begin{equation}
\label{eq-div-stress}
   \stress^{ij}_{\phantom{ij};j} 
       = \stress^{ij}_{\phantom{ij},j}
           + \Gamma^{i}_{kj} \stress^{kj} + \Gamma^{k}_{kj}\stress^{ij}
       = g^{-1/2}\big(g^{1/2} \stress^{ij}\big)_{,j} + \Gamma^{i}_{kj} \stress^{kj}.
\end{equation}
The last equality follows because $(\log g)_{,j}=2\Gamma^{k}_{kj}$.

The differential geometric formulation of elastodynamics, which we just recalled,
does not depend on the choice of a particular coordinate system.
Elasticities, strains, stresses, and displacements transform as tensors
under changes of coordinates.
No generality is lost when we assume that the surface $S$ agrees,
near some chosen point, with a coordinate plane, e.g., $x^3=0$.
In the following, we always use coordinates which are adapted to $S$ in the
following sense:
The normal coordinate $x^3$ is the signed distance to $S$ (such that $x^3<0$ in the exterior),
and the horizontal coordinates $x^1$ and $x^2$ are constant along the geodesics
which intersect the level surfaces of $x^3$ orthogonally.
The horizontal coordinates are completely determined by their restrictions to $S$.
The metric tensor satisfies
$G^{33}=1$ and $G^{3\lambda}=G^{\lambda 3}=0$ if $\lambda<3$.
The surface area element equals $\intd S=\sqrt{g}\intd x^1\intd x^2$, and
$\traction^i\equiv \stress^{i3}$ is the $i$-th component
of the traction  $\traction$ at a given level surface $x^3=$constant.

In the study of surface waves, it is useful to separate, in the
elastodynamic equations and in the formula for the traction,
differentiations in normal direction from differentiations in the
horizontal (and time) directions.
We write
\begin{align}
\label{eqDivStressInNormalCoord}
\stress^{ij}_{\phantom{ij};j}
  &= C^{i3k3}u_{k,33} +(C^{i3k\lambda}+C^{k3i\lambda})u_{k,\lambda 3} 
              +C^{i\lambda k\mu}u_{k,\lambda \mu} 
              +B^{ik\ell}u_{k,\ell} +\tilde B^{ik}u_k, \\
\label{eqTractionInNormalCoord}
\stress^{i3}
 &= C^{i3k3}u_{k,3} + C^{i3k\lambda}u_{k,\lambda} - C^{i3j\ell}\Gamma^k_{j\ell} u_k,
\end{align}
where the implied summations are restricted to $\lambda,\mu <3$.
The coefficients of the first order derivatives are
\begin{gather*}
B^{ik\ell} = C^{ijk\ell}_{\phantom{ijk\ell},j}-C^{i\ell mn}\Gamma^k_{mn} 
   + C^{jmk\ell} \Gamma^i_{jm} + C^{ijm\ell}\Gamma^{m}_{mj}.
\end{gather*}
We have no need to know the $\tilde B^{ik}$'s explicitly,
only their vanishing if the metric tensor $G$ is constant.

\section{Surface waves in homogeneous half-space}
\label{section-homog-halfspace}

In this section we assume that the elastic solid is homogeneous
and occupies a half-space in Euclidean space.
Thus the density, the elasticities, and the metric tensor are constant.
In particular, the Christoffel symbols are zero,
and covariant derivatives are ordinary derivatives.

We use adapted coordinates.
Then $x^3\geq 0$ corresponds the half-space filled by the elastic body.
The unit normal field $\nu$ at $S$, which points into the interior,
has the components $\nu_1=\nu_2=0$, $\nu_3=1$.

The equations of the elastodynamic free boundary problem are as follows:
\begin{gather}
\label{elastodyn-homog}
\rho \ddot u_i 
  - C_i^{\phantom{i}3k3}u_{k,33}
  -(C_i^{\phantom{i}3k\lambda}+C_{\phantom{k3}i}^{k3\phantom{i}\lambda})u_{k,\lambda 3} 
              -C_i^{\phantom{i}\lambda k\mu}u_{k,\lambda \mu}  =0,
\\
\label{elastodyn-homog-free-bdry}
C_i^{\phantom{i}3k3}u_{k,3} + C_i^{\phantom{i}3k\lambda}u_{k,\lambda} =0
\quad \text{at $x^3=0$.}
\end{gather}
Again summation is restricted to $\lambda,\mu <3$.
We have lowered the index $i$ in order that the elastodynamic and the traction
operators map from covariant to covariant components.

Interior and reflected plane waves are found with the time-harmonic ansatz
\[u(x,t)=\exp(\ii\wavenumber(\xi x- c t))U,\qquad \xi x\equiv\xi_jx^j,\quad \ii\equiv\sqrt{-1},\]
where $\wavenumber>0$ is the wave number, $c>0$ the wave speed,
$\xi=[\xi_j]$ the unit propagation direction,
and $\xi/c$ the slowness vector, $|\xi|=1$.
The ansatz satisfies the elastodynamic system if and only if the amplitude vector
$U=[U_k]$ lies in the null space of the acoustic tensor
$[C_i^{\phantom{i}jk\ell}\xi_j\xi_\ell -c^2\rho\delta_i^{\phantom{i}k}]$.
In the special case of isotropy there are two wave speeds,
$c_S=\sqrt{\mu/\rho}$ and $c_P=\sqrt{(\lambda+2\mu)/\rho}$,
the speeds of shear and pressure waves, respectively.
To obtain waves which propagate along the surface, the foregoing ansatz is modified
by allowing $\xi^3$ to be complex with positive imaginary part.

Let $\eta=[\eta_j]$ be a real horizontal unit vector, which means that $|\eta|=1$
and $\eta$ is orthogonal to $\nu$, i.e., $\eta_3=0$.
Consider the acoustic tensor at $\eta+s\nu$:
\[A(s)\equiv s^2 A_0 +s(A_1+A_1^T)+ A_2 -c^2\rho I.\]
Here $I=[\delta_i^{\phantom{i}k}]$ denotes the $3\times 3$ unit matrix, and
\begin{equation}
\label{defAj}
 A_0\equiv [C_i^{\phantom{i}3k3}],\quad A_1\equiv [C_i^{\phantom{i}3k\ell}\eta_\ell],
    \quad A_2\equiv [C_i^{\phantom{i}jk\ell}\eta_j\eta_\ell]
\end{equation}
are real $3\times 3$ matrices;
$A_1^T$ is the transpose of $A_1$ with respect to the inner product defined by $G$.
Because of the symmetries of the stiffness tensor $C$,
the matrices $A_0$, $A_2$, and, for real $s$, $A(s)$ are symmetric with respect to $G$.
The wave speed $c> 0$ is said to be subsonic if $A(s)$ is non-singular for every real $s$.
By the strong convexity of $C$, $A_0$ is positive definite,
and there exists a positive limiting wave speed
$c_\infty=c_\infty(\eta)$ such that $c$ is subsonic if and only if $c<c_\infty$.
Notice that $A(s)$ is positive definite for real $s$ if $c$ is subsonic.

We use the following time-harmonic ansatz to find surface waves 
which have subsonic horizontal slowness $\eta/c$:
\begin{equation}
\label{ray-term-ansatz}
u(x,t) =  \exp(\ii\wavenumber(\eta x-ct)) U(\wavenumber x^3).
\end{equation}
The amplitude $U(z)$ at $z\equiv \wavenumber x^3\geq 0$ is to be determined.
We introduce the differential operator
$DU(z)\equiv -\ii\frac{d}{dz} U(z)$.
The ansatz \eqref{ray-term-ansatz} satisfies equation \eqref{elastodyn-homog} if and only if
\begin{equation}
\label{eq-diffeq-for-ansatz}
A_0D^2U+(A_1+A_1^T)DU+(A_2-c^2\rho I)U = 0
\end{equation}
holds.

The matrix polynomial $A(s)$ satisfies the assumptions
of Proposition~\ref{prop-factor-quad-pol} in the appendix.
Therefore, there is a unique complex $3\times 3$ matrix $Q$
with spectrum in the complex upper half-plane and such that,
with $Q^*$ denoting the adjoint of $Q$,
\begin{equation}
\label{factorization}
A(s) =(sI-Q^*)A_0(sI-Q)
\end{equation}
holds for complex $s$.
Now \eqref{eq-diffeq-for-ansatz} can be rewritten as
\begin{equation}
\label{homog-halfspace-factorization}
(DI-Q^*)A_0(DI-Q)U=0. 
\end{equation}
An analogous formula holds with $Q$ replaced by a matrix $P$,
which has its spectrum in the lower complex half-plane.
The solutions of $DU=PU$ and $DU=QU$ span the space of
solutions of \eqref{eq-diffeq-for-ansatz}.
The solutions of \eqref{eq-diffeq-for-ansatz},
which stay bounded as $\wavenumber x^3\to+\infty$,
are precisely the solutions of $DU=QU$, which are given by
$U(z)=\exp(\ii z Q)U(0)$.
Moreover, as $z\to+\infty$, these solutions decay exponentially.

\begin{remark}
In the argument above, we cannot take $P=Q^*$.
In fact, if \eqref{factorization} and $A(s) =(sI-Q)A_0(sI-Q^*)$ were to hold simultaneously,
$Q$ would be normal, $QQ^*=Q^* Q$.
However, for isotropic media an explicit calculation shows that $Q$ is not normal.
\end{remark}

The surface traction of the waves \eqref{ray-term-ansatz} just constructed equals
\[\traction = \ii\wavenumber\big(A_0DU(0)+ A_1U(0)\big)=-\wavenumber Z U(0).\]
Here $Z$ is the $3\times 3$ surface impedance tensor,
\begin{equation}
\label{def-Z}
Z\equiv -\ii \big(A_0 Q + A_1\big),
\end{equation}
first introduced by \zitatau{IngebTonning69}{Ingebrigtsen and Tonning}.
Summarizing, we have found that, precisely when $ZU(0)=0$ holds,
the ansatz \eqref{ray-term-ansatz} leads to a solution
of \eqref{elastodyn-homog} and \eqref{elastodyn-homog-free-bdry},
which decays into the interior.
Since $Q$ depends smoothly on the subsonic wave speed $c$ and on $\eta$, so does $Z$.
Suppressing the dependence on $\eta$, we often write $Z(c)$ for the impedance tensor.

In the early 1970's the problem of uniqueness and existence of subsonic
surface waves in homogeneous half-space was completely solved.
The following theorem, which we reprove in the next section,
states the solution in terms of the impedance tensor.

\begin{theorem}[Barnett \& Lothe]
\label{Barnett-Lothe-Theorem}
For a given propagation direction, there exists at most one subsonic free surface wave.
A subsonic surface wave exists, if and only if $\det Z(c)<0$
holds for some $c$ in the range $0<c<c_\infty$.
The wave speed $c_R$ is the unique solution, in the range $0<c<c_\infty$,
of the secular equation $\det Z(c)=0$.
The surface displacement of a subsonic surface wave belongs to the
null-space of $Z(c_R)$.
The null-space of $Z(c_R)$ is one-dimensional and contains no non-zero real vector.
\end{theorem}

\section{The surface impedance tensor}
\label{sect-surface-impedance-tensor}

The surface impedance tensor method of subsonic surface wave theory,
developed by \zitatau{IngebTonning69}{Ingebrigtsen and Tonning}
and \zitatau{LotheBarnett85surfwaveimped}{Barnett and Lothe},
relies on the following properties of $Z$.
\begin{proposition}
\label{Z-prop}
For $0\leq c<c_\infty$ the following hold:
\begin{compactenum}[(a)]
\item\label{Z-prop-Hermitian} $Z(c)$ is Hermitian.
\item\label{Z-prop-at-zero-posdef} $Z(0)$ is positive definite.
\item\label{Z-prop-realpart-posdef} The real part of $Z(c)$ is positive definite.
\item\label{Z-prop-atmost-one-nonpos-ev} At most one eigenvalue of $Z(c)$ is non-positive.
\item\label{Z-prop-deriv-posdef} The derivative $\intd Z(c)/\intd c$ is negative definite.
\item\label{Z-prop-limit} The limit $Z(c_\infty)\equiv\lim_{c\uparrow c_\infty} Z(c)$ exists.
\end{compactenum}
\end{proposition}

Theorem~\ref{Barnett-Lothe-Theorem} follows from Proposition~\ref{Z-prop}
combined with arguments of the previous section.
Indeed, it follows from the proposition that the derivative of the
determinant with respect to $c$ is negative at its zeros.
Hence $\det Z(c)=0$ has at most one zero.
Since $\det Z(0)>0$, a zero $c=c_R$ exists if and only if the determinant
becomes negative for some subsonic $c$.
The assertions about the null-space of $Z(c_R)$ follow from
\eqref{Z-prop-realpart-posdef} and \eqref{Z-prop-atmost-one-nonpos-ev}.

Notice that $\det Z(c)<0$ holds if and only if $c_R<c<c_\infty$.
\zitatauthm{LotheBarnett85surfwaveimped}{Barnett and Lothe}{Theorem 12} state an existence criterion
in terms of $Z(c_\infty)$.

\begin{proof}[Proof of Proposition~\ref{Z-prop}]
Following \zitatau{MielkeFu04surfwavespeed}{Mielke and Fu}, we use the Ricatti equation
\begin{equation}
\label{Ricatti-eqn-for-Z}
(Z-\ii A_1^T)A_0^{-1}(Z+\ii A_1)= A_2-c^2\rho I.
\end{equation}
Observing that $Q=\ii A_0^{-1}(Z+\ii A_1)$,
\eqref{Ricatti-eqn-for-Z} is seen to be equivalent to
\[A_0Q^2+(A_1+A_1^T)Q+A_2-c^2\rho I=0,\]
which is the solvency equation \eqref{fact-pos-sa-solvency}
of the factorization \eqref{factorization}.

Passing to adjoints, one recognizes that \eqref{Ricatti-eqn-for-Z}
also holds when $Z$ is replaced by $Z^*$.
Subtracting the two Ricatti equations, one obtains
\[Q^*(Z-Z^*)-(Z-Z^*)Q=0.\]
Because $Q$ and $Q^*$ have disjoint spectra, this Sylvester equation is non-singular.
Hence $Z^*=Z$, proving \eqref{Z-prop-Hermitian}.

To prove \eqref{Z-prop-deriv-posdef},
differentiate \eqref{Ricatti-eqn-for-Z} with respect to $c$, and get
\[\ii Q^*\dot Z - \ii \dot Z Q =-2c\rho I, \quad \dot Z\equiv \intd Z/\intd c.\]
This Lyapunov--Sylvester equation has the unique solution
\[\dot Z= -2c\rho \int_0^\infty \exp(-\ii sQ^*)\exp(\ii sQ)\intd s,\]
which is negative definite.

Next we prove the positive definiteness of $Z(0)=[Z^{k\ell}]$.
Assume that $w$ is a complex vector such that $Z^{k\ell}w_k\overline{w_\ell}\leq 0$.
We have to show that $w=0$.
The exponentially decaying solution of \eqref{eq-diffeq-for-ansatz} with
surface displacement $e^{\ii\eta x} w$ is
\[u =  e^{\ii\eta x} \exp(\ii x^3 Q)w.\]
Denote by $\strain=[\strain_{k\ell}]$, $\strain_{k\ell}=(\partial_\ell u_k+\partial_k u_\ell)/2$,
the strain tensor and by $\traction=- Z(0)w$ the surface traction
of the (complex) displacement field $u$.
Integrate the divergence
\[
\partial_3 (C^{ijk3}\strain_{ij}\overline{u_k})
  = \partial_\ell (C^{ijk\ell}\strain_{ij}\overline{u_k})
  = C^{ijk\ell}\strain_{ij}\overline{\strain_{k\ell}}
\]
over the half-line $x^3\geq 0$, and get the energy identity
\[
\int_0^\infty C^{ijk\ell}\strain_{ij}\overline{\strain_{k\ell}}\intd x^3
  = - \traction^k \overline{w_k}.
\]
By assumption, the right-hand side is non-positive.
Hence, in view of \eqref{elasticities-strong-convex},
the strain vanishes on the half-line $x^3\geq 0$.
In particular,
\[
\strain_{11} =\ii \eta_1 u_1, \quad\strain_{12} =\ii (\eta_1 u_2+\eta_2 u_1)/2,
\quad \text{and}\quad \strain_{33} =\ii (Qu)_3\]
all vanish.
Without loss of generality, we assume that $\eta_1\neq 0$.
It follows that $u_1=u_2=0$.
Moreover,
$\partial_3 |u_3|^2 =2\RE \big(\ii (Qu)_3\overline{u_3}\big)=0$. 
Since $u_3$ tends to zero as $x^3\to\infty$, this implies $u_3=0$.
Therefore $w=0$, which proves assertion \eqref{Z-prop-at-zero-posdef}.

Denote by $A(s)$ the polynomial \eqref{factorization}.
The integral formula for $Q$, \eqref{fact-pos-sa-integral-rep} of the appendix, implies
\[\ii Z\oint_{\gamma_R} A(s)^{-1}\intd s = \oint_{\gamma_R} (sA_0+A_1) A(s)^{-1}\intd s,\]
if the closed contour $\gamma_R$ consists, with $R>0$ sufficiently large,
of the interval $[-R,R]$ and the arc $Re^{i\varphi}$, $0\leq \varphi\leq \pi$.
Notice that $A_0A(s)^{-1}=s^{-2}I + \bigoh(s^{-3})$ as $|s|\to\infty$.
Therefore, letting $R\to\infty$, we obtain
\begin{equation}
\label{Z-integral-formula}
\ii Z\int_{-\infty}^\infty A(s)^{-1}\intd s
   = \pi\ii I +\int_{-\infty}^\infty (sA_0+A_1) A(s)^{-1}\intd s,
\end{equation}
the integral on the right being a principal value integral.
The integrals are real matrices.
Moreover, the integral on the left-hand side of \eqref{Z-integral-formula}
is positive definite.
This proves the assertion \eqref{Z-prop-realpart-posdef}.

If \eqref{Z-prop-atmost-one-nonpos-ev} were not true,
then, for some velocity $c$, the impedance tensor $Z(c)=[Z^{ij}]$ would have
only one positive eigenvalue.
Furthermore, its eigenspace would be one-dimensional.
Since space is three-dimensional,
we could then choose a real vector $w\neq 0$ orthogonal to this eigenspace.
But then $Z^{ij}w_iw_j\leq 0$, contradicting the positive definiteness
of the real part of $Z(c)$.
Hence \eqref{Z-prop-atmost-one-nonpos-ev} holds.

To prove \eqref{Z-prop-limit} we first derive a bound on the operator norm $\|Z(c)\|$ of $Z(c)$.
Notice that $\|Z(c)\|$ equals the maximum of the absolute values of the eigenvalues of $Z(c)$.
By \eqref{Z-prop-deriv-posdef}, $Z(c)\leq Z(0)$ holds with respect to the usual
ordering of Hermitian matrices by the cone of positive definite matrices.
Therefore the eigenvalues of $Z(c)$ are not greater than $\|Z(0)\|$.
The sum of the eigenvalues of $Z(c)$ is positive because it is equal to the
trace of the positive definite matrix $\RE Z(c)$.
Therefore the modulus of a negative eigenvalue, if it exists, is less than the
sum of two positive eigenvalues.
Hence we have $\|Z(c)\|\leq 2\|Z(0)\|$.
By compactness, there exist limit points $Z_*$ of $Z(c)$ as $c\uparrow c_\infty$.
Using \eqref{Z-prop-deriv-posdef} again, we find that $Z_*\leq Z(c)$.
It follows that all limit points are equal, which implies assertion \eqref{Z-prop-limit}.
\end{proof}

\begin{remark}
Formula \eqref{Z-integral-formula} is a variant of the Barnett--Lothe
integral representation.
In fact, passing to the adjoint of \eqref{Z-integral-formula}, and making a
substitution $s=\tan(\phi)$ to replace the integration variable $s$ by an angle $\phi$,
one rederives the formula (2.19) of \zitatau{LotheBarnett85surfwaveimped}{Barnett and Lothe}.
\end{remark}

\begin{remark}
If the coefficients of the matrix polynomial depend smoothly on parameters,
then also the wave speed $c_R$ depends smoothly on these parameters.
Since the zeros of the determinant are simple,
this follows from the implicit function theorem.
\end{remark}

\section{Asymptotic expansions and \pd\ operators}
\label{sect-asymptotic-psido}

If the elastic body is inhomogeneous and the boundary surface curved,
we do not expect that an exact operator factorization \eqref{homog-halfspace-factorization} holds.
However, using \pd\ operators, we shall factorize the elastodynamic operator up to
negligible errors, and we shall use this to exhibit asymptotic subsonic surface waves.
This will be detailed in the following sections.
As a preparation, we summarize standard results on asymptotic expansions and on \pd\ operators.
Refer to \zitatau{Hormander65psido}{H{\"o}rmander}, \zitatauthm{Hormander85anaThree}{H{\"o}rmander}{18.1},
or \zitatau{AlinhacGerard07psido}{Alinhac and G{\'e}rard} for expositions of the basic \pd\ calculus.

We consider vector-valued functions $u(x;\frequency)$ of points $x$ in $d$-dimensional space
which oscillate rapidly as the frequency parameter $\frequency$ tends to $+\infty$.
More precisely,
\begin{equation}
\label{u-asympt-ansatz}
u(x;\frequency) \asymptotic e^{\ii\frequency\theta(x)}
      \sum\nolimits _{k=0}^\infty (\ii\frequency)^{-k} U_{-k}(x),
\end{equation}
where the series is an asymptotic series in the space of smooth functions $x$.
(Often we simply write $=$ instead of $\asymptotic$,
although the equation may be true only in the sense of asymptotic expansions.)
The phase function $\theta(x)$ is real-valued with derivative $\theta'(x)\neq 0$.
We use the operators $D_j\equiv -\ii \partial_j$, where
$\partial_j\equiv \partial/\partial x^j$ denotes partial derivative with
respect to the $j$-th coordinate, $x^j$.
By the Leibniz' product rule, $f\equiv D_j u$ has the asymptotic expansion
\begin{equation}
\label{FundAsymptExpans}
f(x;\frequency) \asymptotic \frequency^m
   e^{\ii\frequency\theta(x)}\sum\nolimits _{k=0}^\infty (\ii\frequency)^{-k} F_{-k}(x)
\end{equation}
with $m=1$ and top order coefficient $F_0(x)=\partial_j\theta(x)U_0(x)$.
A linear differential operator order $m$, $\hat A$,
is a polynomial in $D_j$ of order $m$ with smooth coefficients.
Also $f\equiv \hat A u$ has an expansion \eqref{FundAsymptExpans}.

The \pd\ calculus assigns to certain functions $A(x,\xi)$, which are
defined on $2d$-dimensional phase space, operators $\hat A=A(x,D)$ by
\begin{equation}
\label{KohnNirenberg-quantiz}
(\hat Au)(x)=(2\pi)^{-d} \iint e^{i(x-y)\cdot\xi} A(x,\xi)u(y)\intd y \intd \xi.
\end{equation}
Here $x$ and $\xi$ denote position and (generalized) momentum variables, respectively.
By the Fourier inversion formula, $\hat A$ is the identity when $A=1$.
The function $A(x,\xi)$ is called the (full) symbol of the operator $\hat A$.
We assume that the symbol $A$ belongs to a standard class $S^m=S^m_{1,0}$
of symbols of order at most $m$.
Moreover, the symbol $A$  admits, as $|\xi|\to\infty$,
an asymptotic expansion in homogeneous functions:
\[
  A(x,\xi)\asymptotic\sum \nolimits _{j=0}^\infty A_{-j}(x,\xi), 
  \quad\text{$A_{-j}(x,t\xi)=t^{m-j}A_{-j}(x,\xi)$ for $t>0$.}
\]
(Strictly speaking, homogeneous functions are, unless they are polynomials, not symbols. 
This technicality is overcome by multiplying with cutoff functions
which insure smoothness of symbols at $\xi=0$.)
The symbol $A(x,\xi)$ and the associated operator $\hat A$ are said to be
of order $m$ if the principal symbol $A_0(x,\xi)$ is not identically zero.
Symbols may be matrix-valued.
Rather than using $A_{-1}(x,\xi)$ it is preferable to work
with the subprincipal symbol $A_{\sub}(x,\xi)$,
\begin{equation}
\label{def-Asub}
A_{\sub}\equiv A_{-1} -\frac{1}{2i}\sum\nolimits _j \partial^2 A_0/\partial x^j\partial\xi_j.
\end{equation}
A \pd\ operator is a differential operator if and only if its
symbol is a polynomial in the $\xi$ variable.
In particular, the symbol of $D_j$ is $\xi_j$.

Every \pd\ operator of order $m$ has the property that it maps an
asymptotic sum \eqref{u-asympt-ansatz} into an asymptotic sum \eqref{FundAsymptExpans}.
Moreover, this property is characteristic of \pd\ operators; see \zitatau{Hormander65psido}{H{\"o}rmander}.
Since we need explicit expressions for $F_0$ and $F_{-1}$ we indicate how
\eqref{FundAsymptExpans} is proved using the method of stationary phase.
The method is applied to the summands of \eqref{u-asympt-ansatz}:
\[
e^{-\ii\frequency\theta(x)} (\hat A Ue^{\ii\frequency\theta})(x)
   =(\frequency/2\pi)^{d} \iint
       e^{\ii\frequency(x-y)\cdot\xi-\ii\frequency(\theta(x)-\theta(y))}
          A(x,\frequency\xi)U(y)\intd y \intd \xi.
\]
The stationary point of the phase
\[\Phi\equiv(x-y)\cdot\xi-(\theta(x)-\theta(y)),\quad \Phi_y'=0=\Phi_{\xi}',\]
is at $y=x$, $\xi=\theta'(x)$.
The inverse of the Hessian matrix $H$ of $\Phi$ at the stationary point is
the $2d\times 2d$-matrix
\[
{H(x)}^{-1}=-\begin{bmatrix} 0&I\\ I& \theta''(x)\end{bmatrix},
\quad\text{where $H(x)\equiv \Phi''|_{y=x,\xi=\theta'(x)}$.}
\]
On a formal level, the stationary phase expansion is given by
\[
e^{-\ii\frequency\theta(x)} (\hat A Ue^{\ii\frequency\theta})(x)
  \asymptotic \exp\bigg((\ii \frequency)^{-1}\big(\partial_y\cdot\partial_\xi
      +\frac 1 2 \theta''\partial_y\cdot\partial_y\big)\bigg) (e^{\ii\frequency\rho} U),
\]
where the exponential of differential operators has to be replaced by its formal
Taylor series to give the asymptotic expansion
\zitatthm{Hormander90anaOne}{Theorem 7.7.5}.
Here $\rho$ is the remainder of the second order Taylor expansion
of the phase, it vanishes to third order at the stationary point
where the expressions have to be evaluated.
Thus the principal coefficients in \eqref{FundAsymptExpans} are
\[F_0(x)=A_0(x,\theta'(x))U_0(x),\]
and, suppressing the arguments $x$ and $\xi=\theta'(x)$ in writing,
\[
F_{-1} = \sum_j \frac{\partial A_0}{\partial\xi_j} \frac{\partial U_{0}}{\partial x^j}
            +\frac 1 2 \bigg(\sum_{j,\ell} \frac{\partial^2\theta}{\partial x^j \partial x^\ell}
                       \frac{\partial^2 A_0}{\partial \xi_j \partial \xi_\ell}\bigg) U_{0} 
          +iA_{-1} U_{0} + A_0 U_{-1}.
\]
Observe that the double sum enclosed in parenthesis equals
\[
\sum_{\ell} \frac{\partial}{\partial x^\ell}
   \bigg(\frac{\partial A_0}{\partial \xi_\ell}(x,\theta'(x))\bigg)
  - \sum_{j,\ell} \frac{\partial^2 A_0}{\partial \xi_j \partial x^\ell}(x,\theta'(x)).
\]
Therefore,
\begin{equation}
\label{F-minus-Eins}
F_{-1}= \sum_j \frac{\partial A_0}{\partial\xi_j} \frac{\partial U_0}{\partial x^j}
            +\frac 1 2
\sum_{\ell} \frac{\partial}{\partial x^\ell} \bigg(\frac{\partial A_0}{\partial \xi_\ell}\bigg)U_0
          +iA_{\sub} U_0 +A_0 U_{-1}.
\end{equation}
For $k=2,3,\dots$ there are formulas for $F_{-k}$ which differ from the formula
for $F_{-1}$ by a shifted index and by an additional term which is a sum 
of (derivatives of) $U_{-j}$ with $j<k-1$.
The additional term arises from the higher order terms in the stationary phase expansion.
We emphasize that the asymptotic expansions and the formulas for the
coefficients $F_{-k}$ hold for a general system $\hat A$ of \pd\ operators.

The composition $\hat C=\hat B\hat A$ of \pd\ operators $\hat A$ and $\hat B$
is again a \pd\ operator.
There is a formula for the asymptotic expansion of the full symbol $C(x,\xi)$.
On the principal symbol level, symbols multiply: $C_0=B_0A_0$.
On the next level,
\begin{equation}
\label{eq-psido-compos-sub}
   C_{-1} =B_0A_{-1}+B_{-1}A_0-i\sum\nolimits _j (\partial_{\xi_j} B_0)(\partial_{x^j} A_0)
\end{equation}
holds.
Define the Poisson bracket
\[\{P,Q\}\equiv \sum\nolimits _j 
  (\partial P/\partial\xi_j)(\partial Q/\partial x^j)
 -(\partial P/\partial x^j)(\partial Q/\partial \xi_j),
\]
where the symbols need not be scalar but can take square matrices
(of equal dimensions) as their values.
It follows that the composition of operators is given on the principal and
on the subprincipal level by
\begin{equation}
\label{psido-symbol-compo}
 C_0=B_0A_0,
   \quad C_{\sub} =B_0A_{\sub}+B_{\sub}A_0+\frac{1}{2i}\{B_0,A_0\}.
\end{equation}
The (formal) adjoint of $\hat A$ with respect to a given scalar product
is also a \pd\ operator, $\hat B$, and its principal and subprincipal symbol are given as follows:
\[B_0=A_0^*, \quad B_{\sub} = A_{\sub}^*. \]
The star denotes the (Hermitian) adjoint.
(The symbol formulas arise more naturally if Weyl quantization is used
instead of the Kohn--Nirenberg quantization, $A\mapsto\hat A$ as
in \eqref{KohnNirenberg-quantiz}, which we are using here.)

Under ellipticity assumptions,
one constructs inverses, square roots, and powers of \pd\ operators
which are again \pd\ operators.
The error or remainder terms in such constructions are typically negligible operators
with symbols belonging to $S^{-\infty}\equiv\cap_m S^m$.
In particular, two \pd\ operators with symbols, which are equal when $|\xi|>1$,
differ only by a negligible operator.
Negligible operators map distributions into smooth functions.
If $u$ has an symptotic expansion \eqref{u-asympt-ansatz} and if $\hat A$
is a \pd\ operator which is negligible in conic neighbourhood of the
set of $(x,\frequency\theta'(x))$, $\frequency>0$, then the asymptotic expansion
of $f=\hat A u$ is trivial, i.e., all coefficients $F_{-k}=0$.

The operator $\hat L$ of elastodynamics is a second order differential $3\times 3$ system,
\begin{equation}
\label{defOpElastodyn}
(\hat L u)_i\equiv \rho {\ddot u}_i - \stress^{\phantom{i}j}_{i\phantom{j};j}. 
\end{equation}
The principal symbol $L_0$ of $\hat L$ is the acoustic tensor,
\begin{equation}
\label{princsymbElastodynOp}
L_0(x,\xi)=\big[C_i^{\phantom{i}jk\ell}(x)\xi_j\xi_\ell-\rho\xi_0^2\delta_i^{\phantom{i}k}(x)\big].
\end{equation}
Here we included the time coordinate as $x^0=t$, and the dual variable (frequency) as $\xi_0$.
The implied summation extends over the spatial indices, i.e., $j,\ell \geq 1$.
The \pd\ operator which will be most important to us is the surface impedance
operator, which, up to negligible error terms,
maps surface displacements to surface tractions of solutions of $\hat L u=0$.
This operator is introduced in Section~\ref{sect-exist-freesurfwave}.

\section{Ray theory for systems}
\label{section-ray-theory}

The ray method solves a wave equation $\hat Au=0$
by solving the equations $F_{-j}=0$ for the amplitudes $U_{-k}$
of an asymptotic expansion \eqref{u-asympt-ansatz} of $u$.
For $j>0$, the equations $F_{-j}=0$ are ordinary differential equations along rays,
commonly called transport equations.

The equation $F_0=0$ is the dispersion equation, $A_0(x,\theta'(x))U_0(x)=0$.
If $A_0$ is scalar, then this becomes the eikonal equation $A_0(x,\theta'(x))=0$.
Furthermore, $F_{-1}=0$ with \eqref{F-minus-Eins} is a transport equation for $U_0$
which is solved by the method of characteristics.
If $A_0$ is not scalar, then it is not obvious how \eqref{F-minus-Eins} and
$F_{-1}=0$ lead to a useful transport equation for $U_0$.

For systems of real principal type, introduced by \zitatau{Dencker82polarsets}{Dencker},
there is an efficient ray theory based on the theory of Fourier integral operators.
It applies to isotropic elastodynamics, \zitatau{SHRoehrig04lagrsol}{Hansen and R{\"o}hrig}, and, as we will show,
to the propagation of Rayleigh waves along curved surfaces of general elastic media.
We adapt this ray method to our setting.

Let $\hat A$ be a (square) system of \pd\ operators with principal symbol $A_0(x,\xi)$.
The set or zeros of the determinant $\det A_0$ is called the characteristic set of
$\hat A$ (or of $A_0$).
We assume that the system $\hat A$ is of real principal type in the sense
of \zitatau{Dencker82polarsets}{Dencker}.
This means that there exists a matrix-valued symbol $B(x,\xi)$,
homogeneous of degree zero in the $\xi$ variable,
and a scalar symbol $a(x,\xi)$ such that
$B(x,\xi)A_0(x,\xi)=a(x,\xi)I$ holds with $I$ denoting the unit matrix.
Moreover, $a$ is real-valued and its set of zeros equals the
characteristic set, on which it vanishes simply,
i.e.\ $\partial_\xi a(x,\xi)\neq 0$ holds whenever $a(x,\xi)=0$.
We call $a$ a Hamilton function of $A_0$ (and of the operator $\hat A$).
In the special case where the determinant of $A_0$ vanishes simply on the characteristic set,
$B$ is a scalar multiple of the cofactor matrix,
and $ab=\det A_0$ with $b$ a nowhere vanishing function.

Next we collect some properties which follow from the
real principal type assumption \zitat{Dencker82polarsets}.
The simple vanishing of $a$ implies that the complement of the
characteristic set is dense in phase space.
Moreover, $BA_0=aI$ and also $A_0B=aI$ holds.
This is clear where $a\neq 0$, and, by continuity, this holds in general.
The following identities between range spaces and null spaces (kernels)
are important:
\begin{equation}
\label{range-eq-kernel}
\range(B)=\ker(A_0),\quad \range(A_0)=\ker(B),\quad \text{where $a=0$.}
\end{equation}
The inclusion of the ranges in the null spaces is clear.
Differentiating $BA_0=aI=A_0B$ in an appropriate direction, with the
directional derivative denoted by a prime, we obtain
\[BA_0'+B'A_0=a'I=A_0'B+A_0B', \quad a'\neq 0.\]
Now, equality in \eqref{range-eq-kernel} follows from rank considerations.

Fix a real-valued solution $\theta(x)$ of the Hamilton--Jacobi equation $a(x,\theta'(x))=0$.
The integral curves $(x(s),\xi(s))$ of Hamilton's canonical equations,
$\dot x=\partial_\xi a(x,\xi)$ and $\dot\xi=-\partial_x a(x,\xi)$,
satisfy $\xi(s)=\theta'(x(s))$ for all parameters $s$ if this is true for at least one $s$.
These integral curves in phase space are called bicharacteristics.
The Poisson bracket $\{a,b\}$ is the derivative of
a function $b(x,\xi)$ in bicharacteristic direction,
\[\frac{\intd}{\intd s} b(x(s),\xi(s))=\{a,b\}(x(s),\xi(s)).\]
The projections of bicharacteristics to space-time
are the space-time rays $x(s)$ associated with $\theta$.
The vector field $V(x)\equiv \partial_\xi a(x,\theta'(x))$ is called the ray field.
We denote by a dot the derivative in ray direction,
\[\dot U(x(s))\equiv (\intd/\intd s)U(x(s))=(V\cdot\partial_x U)(x(s)).\]

Next, we turn to the solution of $F_0=0$ and $F_{-1}=0$.
Again we suppress arguments $x$ and $\xi$, and we assume evaluation
of expressions at $\xi=\theta'(x)$.
If $A_0U_0=0$ holds, then the leading term in the asymptotic expansion
of $\hat B f=\hat B\hat A u$ equals
\[BF_{-1} = \dot U_0+ (\Div  V/2)U_0 + i A_{\sub}'U_0,\]
where $\Div V\equiv \nabla_x\cdot V$ is the divergence of the ray field,
and $A_{\sub}'$ denotes the subprincipal symbol of $\hat A'\equiv \hat B\hat A$.
To see this, observe that the right-hand equals the right-hand side in \eqref{F-minus-Eins}
when $\hat A$ is replaced by $\hat A'$.
Observe that $A_0'=aI$, and, by \eqref{psido-symbol-compo},
$A_{\sub}'=BA_{\sub}+\frac{1}{2i}\{B,A_0\}$ when restricted to the nullspace of $A_0$.
Therefore, 
\[BF_{-1} = \dot U_0+ (\Div  V/2)U_0 +\frac{1}{2}\{B,A_0\}U_0 + iB A_{\sub}U_0.\]
To solve $F_{-1}=0$ we first solve, using the following lemma,
the equations $BF_{-1}=0$ and $A_0U_0=0$, simultaneously.
\begin{lemma}
\label{lemma-soln-transport-eqn}
If $U$ solves the differential equation
\begin{equation}
\label{general-transport-equation}
\dot U+ (\Div  V/2)U +\frac 1 2\{B,A_0\} U + BW=0,
\end{equation}
then $A_0U=0$ holds along a ray if this holds in at least one point of the ray.
\end{lemma}
\begin{proof}
Multiply \eqref{general-transport-equation} by $A_0$ from the left to get
\[A_0\dot U+ (\Div V/2)A_0U +\frac 1 2A_0\{B,A_0\} U=0\]
along rays.
The derivative of $A_0(x,\xi)$ along bicharacteristics is given by
\[
2\frac{\intd}{\intd s} A_0\equiv 2\{aI,A_0\}
    =\{A_0B,A_0\}-\{A_0,BA_0\} = A_0\{B,A_0\}-\{B,A_0\}A_0.
\]
We obtain a homogeneous linear ordinary differential equation for $A_0U$,
\[\frac{\intd}{\intd s}(A_0U)+ (\Div V/2)A_0U +\frac 1 2\{B,A_0\}A_0 U=0.\]
The assertion of the lemma follows from the uniqueness of solutions to
initial value problems of this equation.
\end{proof}

Using the lemma, find a solution $U_0\neq 0$ of the transport equation
\begin{equation}
\label{transport-equation}
\dot U_0+ (\Div  V/2)U_0 +\frac{1}{2}\{B,A_0\}U_0 + iB A_{\sub}U_0 =0
\end{equation}
such that $A_0U_0=0$.
Thus $F_0=0$ and $BF_{-1}=0$.
Recall that the null space of $B$ equals the range of $A_0$.
Therefore, there exists $U_{-1}$ such that $F_{-1}=-A_0U_{-1}$.
Set $u= e^{\ii\frequency\theta}(U_0+ (\ii\frequency)^{-1} U_{-1})$
to obtain $F_0=0$ and $F_{-1}=0$.
To proceede by recursion, assume that
$u= e^{\ii\frequency\theta}\sum_{j<k} (\ii\frequency)^{-j} U_{-j}$
has been found such that $F_{-j}=0$ for $j<k$.
Then
\[
BF_{-k} = \dot U_{-k+1}+ (\Div  V/2)U_{-k+1} + i A_{\sub}'U_{-k+1} +W_{-k},
\]
where $W_{-k}$ is a sum of derivatives of $U_{-j}$, $j<k-1$.
We shall modify $U_{-k+1}$ and add $e^{\ii\frequency\theta} (\ii\frequency)^{-k} U_{-k}$
to $u$ such that $F_{-j}=0$ for $j\leq k$.
Using the lemma, solve
\[
\dot U+ (\Div  V/2)U +\frac{1}{2}\{B,A_0\}U + iB A_{\sub}U + BF_{-k}=0,
\quad A_0 U=0.
\]
Replace $U_{-k+1}$ by $U_{-k+1}+U$.
Then still $F_{-j}=0$ for $j<k$, and, in addition, we now have $BF_{-k}=0$.
Choose $U_{-k}$ as a solution of $F_{-k}+A_0U_{-k}=0$.
With this choice of $U_{-k}$ we arrive at $F_{-k}=0$,
which completes the recursive step of the construction of an asymptotic solution of $\hat A u=0$.

For later reference, we summarize the result just obtained.

\begin{proposition}
\label{prop-ray-theory}
Let $\hat A$ \pd\ system of real principal type with Hamilton function $a$,
and $B(x,\xi)A_0(x,\xi)=a(x,\xi)I$.
Let $\theta(x)$ with $\theta'(x)\neq 0$ be a solution of the eikonal equation $a(x,\theta'(x))=0$.
Denote by $V(x)\equiv \partial_\xi a(x,\theta'(x))$ the ray field,
and $\Div V\equiv \nabla_x\cdot V$ its divergence.
Let $U_0(x)$ be a solution of the transport equation \eqref{transport-equation}
which satisfies
\[A_0(x,\theta'(x))U_0(x)=0. \]
Then there is a solution of $\hat A u=0$ which is given by an
asymptotic sum \eqref{u-asympt-ansatz} with leading amplitude $U_0(x)$.
\end{proposition}

\section{Subsonic displacement boundary problems}
\label{sect-displacement-bp}

We return to elastodynamics.
Given a subsonic phase $\theta$ and a vector-valued amplitude $W$,
each defined and smooth on the space-time boundary of the elastic body,
we solve, asymptotically as $\frequency\to\infty$, the displacement boundary problem
\[\hat Lu=0\quad\text{in $B$,} \qquad u =e^{\ii\frequency\theta}W\quad\text{at $S$}.\]
To construct $u$ we use coordinates adapted to $S$
as introduced in Section~\ref{section-geometrical-language}.
Thus the boundary condition reads
\begin{equation}
\label{displacementBdryCond}
u|_{x^3=0} = e^{\ii\frequency\theta(x^0,x^1,x^2)}W(x^0,x^1,x^2).
\end{equation}
The time variable is denoted $t$ or $x^0$, the dual variable is $\xi_0$.

First, we define the subsonic region.
Recall from \eqref{defAj} the matrices $A_j$.
Since the material parameters are allowed to vary smoothly with spatial position $x$,
we have $A_j=A_j(x,\eta)$, and
the acoustic tensor \eqref{princsymbElastodynOp} is given as follows:
\[
L_0(x,\xi)= \xi_3^2 A_0(x) +\xi_3(A_1(x,\eta)+A_1(x,\eta)^T)
   + A_2(x,\eta) -\xi_0^2\rho(x) I,
\]
where $\eta=(\xi_1,\xi_2,0)$.
We say that, at the surface point $y=(x^1,x^2)$, the covector
$\zeta=(\xi_0,\xi_1,\xi_2)$ is subsonic if
$L_0(x,\xi)|_{x^3=0}$ is positive definite for real $\xi_3$.
If $\zeta$ is subsonic, then so are $-\zeta$ and $s\zeta$ for $s>0$.
As in Section~\ref{section-homog-halfspace},
we appeal to Proposition~\ref{prop-factor-quad-pol} of the appendix,
and we get, if  $\zeta$ is subsonic at $y$, a factorization
\begin{equation}
\label{factor-of-acoustic-tensor-x-depend}
L_0(x,\xi)=(\xi_3I-Q_0^*(x,\zeta))A_0(x)(\xi_3I-Q_0(x,\zeta)),
\end{equation}
where $Q_0(x,\zeta)$ depends smoothly on $x$ and $\zeta$,
and the spectrum of $Q_0$ is contained in the upper complex half-plane.
The factorization \eqref{factor-of-acoustic-tensor-x-depend} 
holds in a boundary layer $0\leq x^3<\delta$.
The matrices $Q_0(x,\zeta)$ are uniquely determined.
Using the uniqueness, it is seen that $Q_0$ satisfies $Q_0(x,s\zeta)=sQ_0(x,\zeta)$ if $s>0$.
However, note that $Q_0(x,-\zeta)\neq -Q_0(x,\zeta)$ because of the spectral condition.
It follows that $Q_0$ cannot be a polynomial in $\zeta$.

The phase function $\theta$ is assumed to be subsonic.
By definition, this means that its derivatives $\zeta=\theta'$ are non-zero and subsonic.
Moreover, we assume that we have a zeroth order scalar symbol $\chi(x,\xi)$,
supported in the subsonic region, such that, for $|\xi|>1$,
$\chi(x,\xi)$ equals unity in an open conic neighbourhood $\Gamma_1$
of the set $\Gamma$ of $(x,\xi)$, $\xi=s\theta'(x)$, $s>0$.

The \pd\ operator $(\chi Q_0)(x,D')$, $D'\equiv (D_0,D_1,D_2)$,
is defined, has the order one, and is tangential.
A \pd\ operator is said to be tangential if it commutes with multiplication by $x^3$,
or, equivalently, the operator pseudo-differentiates only with respect to the coordinates
$(x^0,x^1,x^2)$, and depends smoothly on $x^3$ as a parameter.

The elastodynamic operator $\hat L$ is a second degree polynomial in $D_3=-\ii\partial_3 I$
with tangential differential operators as coefficients.
Let $\hat Q$ be a tangential first order \pd\ operator,
say $\hat Q=(\chi Q_0)(x,D')$, whose principal symbol is, on $\Gamma_1$, equal to $Q_0$.
By \eqref{factor-of-acoustic-tensor-x-depend} and the symbol calculus we have
\begin{equation}
\label{operatorFactorization}
\hat L= (D_3-\hat P)A_0(D_3-\hat Q) + \hat R^0 D_3 +\hat R^1,
\end{equation}
where $\hat P$ is the adjoint of $\hat Q$, $P_0=Q_0^*$, and $\hat R^j$ is a tangential
\pd\ operator of order at most $j$.
Moreover, on $\Gamma_1$, the principal symbols $R^j_0$ of $\hat R^j$ satisfy
\begin{align*}
R^0_0 &= B_0 +\ii\partial_3 A_0 +A_0Q_{-1} +P_{-1}A_0
      -\ii\sum\nolimits _j (\partial_{\xi_j}P_0)(\partial_{x^j} A_0), \\
R^1_0 &= B_1 -\ii\partial_3 (A_0Q_0)  -P_0A_0Q_{-1} -P_{-1}A_0Q_0
      +\ii\sum\nolimits _j (\partial_{\xi_j}P_0)(\partial_{x^j} A_0Q_0).
\end{align*}
Here we used \eqref{eq-psido-compos-sub} and \eqref{eqDivStressInNormalCoord},
and we have set
\begin{equation}
\label{def-B0-B1}
B_0(x) =[B^{ik3}],\quad B_1(x,\eta)=[B^{ik\lambda}\eta_\lambda]. 
\end{equation}

Our next goal is to improve, by modifying $\hat Q$ and $\hat P$,
the factorization \eqref{operatorFactorization}
so that the remainders $\hat R_j$ are negligible in a conic neighbourhood of $\Gamma$.
To do so, we first eliminate $P_{-1}$ from the above equations for $R^j_0$, and we set $R^j_0=0$.
We derive the following equation for $Q_{-1}$:
\begin{equation}
\label{Lyapunov-AQminus}
\begin{aligned}
Q_0^* & A_0Q_{-1} -A_0Q_{-1}Q_0 \\
 &= B_0Q_0+B_1 - \ii A_0\partial_3 Q_0
   + \ii\sum\nolimits _j (\partial_{\xi_j}Q_0^*)A_0(\partial_{x^j} Q_0).
\end{aligned}
\end{equation}
This is a uniquely solvable Sylvester equation for $A_0Q_{-1}$
because $Q_0$ and its Hermitian adjoint have disjoint spectra.
Having found $Q_{-1}$, we solve for $P_{-1}$.
The symbols $Q_{-1}$ and $P_{-1}$ are only defined in $\Gamma_1$.
Fix a new cutoff symbol $\chi_1$ which equals unity in a conic neighbourhood
of $\Gamma$ and is supported in the set where $\chi=1$.
Replace $\hat Q$ by $\hat Q+(\chi_1 Q_{-1})(x,D')$, similarly for $\hat P$.
Then \eqref{operatorFactorization} holds with remainders $\hat R^j$
which have, where $\chi_1=1$, order at most $j-1$.
This construction of reducing the order of the remainders can be iterated.
In fact, one recursively solves equations $Q_0^*X-XQ_0=Y$ for $X=A_0Q_{-k}$.
Using asymptotic summation, a factorization \eqref{operatorFactorization}
is obtained such that the remainders are negligible in a neighbourhood of $\Gamma$.

\begin{remark}
The result \eqref{operatorFactorization} is a factorization of the elastodynamic operator
into a product of \pd\ operators.
The factorization does not hold everywhere, but only microlocally in
(some chosen subregion of) the subsonic region.
Such a factorization is not possible using only  differential operators,
\end{remark}

Our main application of the factorization \eqref{operatorFactorization}
is the construction of solutions to displacement boundary problems.

\begin{lemma}
\label{lemma-u-expdecaying}
Let $\theta(x^0,x^1,x^2)$ be a subsonic phase function, $W(x^0,x^1,x^2)$ a smooth amplitude.
There is a solution $u$,
\begin{equation}
\label{ansatz-u-expdecaying}
u\asymptotic e^{\ii\frequency\theta(x^0,x^1,x^2)}\sum\nolimits _{j=0}^\infty
      (\ii\frequency)^{-j} U_{-j}(x^0,x^1,x^2,\frequency x^3),
\end{equation}
which solves $\hat Lu=0$ asymptotically as $\frequency\to \infty$ in $x^3\geq 0$,
and satisfies the displacement boundary condition \eqref{displacementBdryCond}.
Moreover, $D_3u =\hat Q u$, and $u$ decays exponentially in $\frequency x^3$.
\end{lemma}

The proof consists of finding an asymptotic solution \eqref{ansatz-u-expdecaying}
of $(D_3-\hat Q)u=0$ in $x^3\geq 0$.
To satisfy the boundary condition \eqref{displacementBdryCond} 
we require $U_0|_{x^3=0}=W$ and $U_{-j}|_{x^3=0}=0$ if $j>0$.
By the asymptotic expansion lemma for \pd\ operators,
\[
\ii(D_3-\hat Q)u\asymptotic  \frequency e^{\ii\frequency\theta}\sum\nolimits _{j=0}^\infty
    (\ii\frequency)^{-j} F_{-j}(x^1,x^2,\frequency x^3),
\]
where, if we set $\zeta=\theta'(x^0,x^1,x^2)$, $y=(x^1,x^2,0)$,
\[
F_{-j}(x) = \partial_3 U_{-j}(x) -\ii Q_0\big(y,\zeta\big) U_{-j}(x) + G_{-j}(x).
\]
When deriving these formulas, the symbols $Q_{-j}$ of $\hat Q$ are replaced
by their Taylor series in the $x^3$ variable, and it is used that a power
of $x^3$ give rise to a power of $1/\frequency$ with the same exponent.
The remainder term $G_{-j}$ is an expression involving $U_{-\ell}$ for $\ell<j$, and $G_0=0$.
We recursively solve the equations $F_{-j}=0$ for $U_{-j}$ imposing,
in addition, the above-mentioned boundary conditions at the surface.
Since $\ii Q_0$ has its spectrum in the left half-plane,
the solutions $U_{-j}(x)$ decay exponentially as $x^3\to\infty$.
The factorization remainder $\hat R^0 D_3+\hat R^1$ is negligible in a
neighbourhood of $\Gamma$.
Therefore, $(D_3-\hat Q)u=0$ implies $\hat L u=0$.

The assertion of the lemma still holds if \eqref{displacementBdryCond} is
replaced by the more general displacement boundary condition $u|_{x^3=0}= w$,
\[
w(x^0,x^1,x^2)\asymptotic 
   e^{\ii\frequency\theta(x^0,x^1,x^2)}\sum\nolimits _{j=0}^\infty
         (\ii\frequency)^{-j}W_{-j}(x^0,x^1,x^2).
\]
In fact, the proof of the lemma is readily adapted to cover this situation too.
\begin{remark}
The methods of this section originate from the theory of elliptic boundary problems.
\zitatau{WlokaRowleyLawruk95}{Wloka et al.} systematically use spectral factorizations
of matrix polynomials in their treatment of such problems.
\end{remark}

\section{Existence of subsonic free surface waves}
\label{sect-exist-freesurfwave}

As in \eqref{def-Z} we define, in the subsonic region,
the surface impedance tensor $Z_0$:
\begin{equation}
\label{def-impedance-tensor}
\ii Z_0(y,\zeta)\equiv A_0(y)Q_0(y,\zeta)+A_1(y,\eta).
\end{equation}
Here $y=(x^1,x^2)=(x^1,x^2,0)$, $\eta=(\xi_1,\xi_2)=(\xi_1,\xi_2,0)$, and $\zeta=(\xi_0,\eta)$.
Since $Q_0$ and $A_1$ are homogeneous of degree $1$ in the variables $\xi_j$, so is $Z_0$.
The subsonic region is conic.
Therefore, there is a function $c_\infty(y,\eta)>0$,
homogeneous of degree one in $\eta$, such that $\zeta$ is subsonic if and only if
$c_\infty(y,\eta) > |\xi_0|$ holds.
If we freeze a point $y$ on the surface and
a unit horizontal propagation direction $\eta$, $|\eta|=1$,
then $Z_0$ is the surface impedance tensor of a homogeneous half-space:
$Z_0(y,\zeta)=Z(c)$, where $c=-\xi_0$,
and $c_\infty(y,\eta)$ is the limiting velocity.
This shows that Proposition~\ref{Z-prop} and the proof of the Barnett--Lothe Theorem
apply mutatis mutandis to $Z_0(y,\zeta)$.

Subsonic free surface waves correspond to solutions of $\det Z_0(y,\zeta)=0$,
which is known as the secular equation.
We assume from now on that the secular equation has a zero $\xi_0<0$,
necessarily unique and simple, for every surface point $y$ and every direction $\eta\neq 0$.
By the implicit function theorem,
there is a smooth function $c_R(y,\eta)$, homogeneous of degree one in $\eta$,
such that $0<c_R(y,\eta)<c_\infty(y,\eta)$, and
\begin{equation}
\label{det-Z0-simple}
\det Z_0=bh,\quad h(y,\zeta)\equiv \xi_0+c_R(y,\eta), \quad b(y,\zeta) > 0.
\end{equation}
We call $h$ the Hamilton function of the free surface wave problem.
The phase function $\theta\equiv\varphi(y)-t$ is subsonic and
solves $h(y,\theta')=0$ if and only if the eikonal
equation $c_R(y,\varphi'(y))=1$ is satisfied.
The ray field is $V=(1,(\partial c_R/\partial\eta)(y,\varphi'(y)))$.
So we take time $t$ as the parameter along rays, and we regard rays
as the spatial curves $y(t)$ which satisfy
\[\frac{\intd}{\intd t}y(t)=\frac{\partial c_R}{\partial\eta}\big(y(t),\varphi'(y(t))\big).\]
Note that the null-space $N(y)$ of $Z_0(y,-c_R(y,\eta),\eta)$, $\eta=\varphi'(y)$,
is one-dimensional.
Denote by $r\geq 0$ the distance to the boundary surface $S$.

\begin{theorem}[Existence of subsonic free surface waves]
\label{main-thm-exist-ssw}
Let $\varphi$ be a solution of the eikonal equation
\begin{equation}
\label{thm-eikonal}
c_R(y,\varphi'(y))=1.
\end{equation}
There is a matrix-valued function $H(y)$, which can be calculated algebraically
from the derivatives up to second order of
the elasticities, the material density, and the surface $S$,
such that the following holds.
Let $W_0(y)$ satisfy, for every ray $y(t)$, the transport equation
\begin{equation}
\label{thm-transport-eqn-Wzero}
\frac{\intd}{\intd t}W_0(y(t))+\frac 1 2\Div V(y(t))W_0(y(t))+ H(y(t))W_0(y(t))=0
\end{equation}
and $W_0(y(t))\in N(y(t))$ for some $t$.
Then $W_0(y)\in N(y)$ holds for all $y$.
The free surface boundary problem \eqref{eq-elastodyn-freebdry} has a non-zero solution
\begin{equation}
\label{thm-u}
u(t,y,r;\omega)\asymptotic e^{\ii\frequency(\varphi(y)-t)}\sum\nolimits _{j=0}^\infty
      (\ii\frequency)^{-j} U_{-j}(y,r\omega).
\end{equation}
The amplitudes $U_{-j}(y,r\omega)$ decay exponentially as $r\omega\to\infty$,
and the leading surface amplitude equals $U_0(y,0)=W_0(y)$.
\end{theorem}

Let $\hat Q$ be the \pd\ operator of Lemma~\ref{lemma-u-expdecaying}
with the defining cutoff symbol $\chi$ chosen equal to unity in a conic
neighbourhood of the set $\Xi$ which consist of $(y,s\zeta)$,
where $\zeta=(-1,\varphi'(y))$, $s>0$. 
Consider
\begin{equation}
\label{def-T}
T=[T_i^k], \quad T_i^k\equiv C_i^{\phantom{i}3j\ell}\Gamma^k_{j\ell},
\end{equation}
and regard $T$ as a multiplication operator.
We define the surface impedance operator $\hat Z$ by
\begin{equation}
\label{defImpedanceOperator}
\hat Z\equiv -\ii(A_0 \hat Q+\hat A_1 +\ii T).
\end{equation}
Notice that $\hat Z$ is a first order \pd\ operator on the space-time boundary surface.
The principal symbol of $\hat Z$ is $-\ii(A_0\chi Q_0+A_1)$
which equals the surface impedance tensor $Z_0$ in a conic neighbourhood of $\Xi$.
Let $B(y,\zeta)$ be the cofactor matrix of $Z_0(y,\zeta)$ divided by the
factor $b(y,\zeta)$ of \eqref{det-Z0-simple}.
By Cramer's rule,
\begin{equation}
\label{rpt-factor-B-of-Z}
B(y,\zeta)Z_0(y,\zeta)=h(y,\zeta)I.
\end{equation}
Hence $\hat Z$ is (microlocally near $\Xi$)
a real principal type operator with Hamilton function $h$.

We can now apply Proposition~\ref{prop-ray-theory} to obtain a solution
\[
w\asymptotic e^{\ii\frequency(\varphi(y)-t)}\sum\nolimits _{j=0}^\infty
     (\ii\frequency)^{-j} W_{-j}(y)
\]
of $\hat Z w=0$.
(Notice that the transport equations can be solved with amplitudes
which do not depend explicitly on time $t$.)
From \eqref{transport-equation} we obtain the transport equation
\eqref{thm-transport-eqn-Wzero}, where the coefficient matrix is given by
\begin{equation}
\label{coeff-H-transport-eqn}
H= \frac{1}{2}\{B,Z_0\} + iB Z_{\sub}
\end{equation}
after evaluating at $\eta=\varphi'(y)$ and $\xi_0=-c_R(y,\eta)$.
Here $B$ is defined in \eqref{rpt-factor-B-of-Z}, and $Z_{\sub}$
is the subprincipal symbol of the surface impedance operator.
Using Lemma~\ref{lemma-soln-transport-eqn} it follows that $W_0(y)\in N(y)$.
Using Lemma~\ref{lemma-u-expdecaying} and the remark at the end of its proof,
we obtain a solution \eqref{thm-u} of the following displacement boundary problem:
$\hat L u=0$ in $r=x^3\geq 0$, and $u=w$ at $x^3=0$.
Moreover, we get that the $U_{-j}$ decay exponentially as $\frequency r\to\infty$.

Recall the formula \eqref{eqTractionInNormalCoord} for the surface
traction $\traction= [\stress_i^{\phantom{i}3}]$ of the field $u$.
Since $D_3 u=\hat Q u$ holds, we find that, in view of the definition 
of the surface impedance operator, the surface traction vanishes:
\[-\ii \traction = A_0D_3 u +\hat A_1 u +\ii Tu =\ii \hat Z w=0\quad\text{at $x^3=0$.}\]
This proves the existence of a free subsonic surface wave.
In the next section we explain how $H(y)$ is evaluated.

\begin{remark}
The theorem says, in particular, that
$c_R(y,\eta)/|\eta|$ is the wave speed of a subsonic free surface wave
at the surface point $y$ in the horizontal direction $\eta=\varphi'(y)$.
Notice that an ansatz \eqref{thm-u} will solve \eqref{eq-elastodyn-freebdry} only if
the phase function $\varphi$ satisfies the eikonal equation \eqref{thm-eikonal}.
\end{remark}

\section{Evaluation of the transport equation}
\label{sect-transport-eqn}

Before we discuss the evaluation of the coefficient matrix $H(y)$,
let us remark how the transport equation \eqref{thm-transport-eqn-Wzero}
can be reduced to a scalar equation.
Assume given a smooth unit field $W(y)$ which spans the one-dimensional null-space $N(y)$.
Then the leading amplitude equals $W_0=\psi W$ with a complex-valued function $\psi(y)$
which satisfies, along rays $y(t)$, the scalar transport equation
\begin{equation}
\label{scalar-transport-eqn}
\frac{\intd}{\intd t} \psi +\frac 1 2 (\Div V)\psi
 +\big\langle W|\frac{\intd}{\intd t}W+ HW\big\rangle\psi=0.
\end{equation}
The angular brackets denote the inner product defined by the metric tensor,
conjugate linear in the first slot.
If the surface $S$ is not plane, then the inner product will
not be constant along rays, in general.

We assume that the elasticity tensor, the material density, and the boundary surface
are known, including first and second order derivatives.
These data determine the transport equation.
We make no attempt to derive analytical solution formulas.
Rather we explain how the transport equation can be set up
ready for numerical computation.
To do actual computations, one chooses coordinates $y=(x^1,x^2)$ on $S$.
The well-known ray coordinates are one convenient choice.
For computations at interior points, the unique extensions
as adapted coordinates, $x=(x^1,x^2,x^3)$, are used.
The data determine the acoustic tensor or principal symbol
of the elastodynamic operator, $L_0$.
At every point of $S$, the acoustic tensor is viewed as a matrix polynomial
in the variable conormal to $S$ at $y$.
A factorization algorithm for matrix polynomials has to produce the
factorizations \eqref{factor-of-acoustic-tensor-x-depend}.
Because of homogeneity it suffices to evaluate $Q_0(y,\xi_0,\eta)$ only at unit vectors $\eta$.
To actually compute $Q_0$ one may have to determine eigenvalues, eigenvectors,
and in general also Jordan-Keldysh chains, of the matrix polynomial.
The result of the computation also gives,
by \eqref{def-impedance-tensor}, the impedance tensor $Z_0$.

First and second order derivatives of $Z_0$ are also needed.
To calculate these,
we procede as in the proof of Proposition~\ref{Z-prop}\eqref{Z-prop-deriv-posdef}.
The surface impedance tensor satisfies the Ricatti equation
\[
\big(Z_0(\zeta)-\ii A_1^T(\eta)\big)A_0^{-1}\big(Z_0(\zeta)+\ii A_1(\eta)\big)
   = A_2(\eta)-\xi_0^2\rho I,
\]
where we suppressed dependency on $x$ from the notation.
Denote differentiation with respect to some chosen coordinate or parameter by a prime.
Differentiate the Ricatti equation, and obtain a linear equation
\begin{equation}
\label{QstarX-XQ-eq-Y}
Q_0^* X-X Q_0=Y
\end{equation}
for the complex $3\times 3$ matrix $X\equiv Z_0'$.
The right-hand side $Y$ is evaluated as an algebraic expression
in $Z_0$, $A_0$, $A_1$, $\rho$, and in the derivatives $A_0'$, $A_1'$, and $\rho'$.
Since $Q_0$ and $Q_0^*$ have disjoint spectra,
the Sylvester equation \eqref{QstarX-XQ-eq-Y} is uniquely solvable.
Knowing $Z_0'$ we also know the derivative $Q_0'$ of $Q_0=A_0^{-1}(\ii Z_0-A_1)$.
Differentiating \eqref{QstarX-XQ-eq-Y}, we obtain an equation for the derivative $X'$:
\[Q_0^* X'-X' Q_0=Y'-(Q_0')^* X- X Q_0'.\]
We conclude that first and second order derivatives of $Z_0$ can be computed 
from those of the elasticities and the material density by purely algebraic computations.

Next we indicate how to evaluate the Poisson bracket term \eqref{coeff-H-transport-eqn}.
Recall that $bB$ is the cofactor matrix of $Z_0$.
Therefore, a first order derivative $(bB)'$ is an algebraic expression
in the elements of $Z_0$ and $Z_0'$.
So we can evaluate $\{bB,Z_0\}$.
Observe from \eqref{det-Z0-simple} that
\begin{equation}
\label{b-eq-detprime}
b(y,\eta)=\partial_{\xi_0}\det Z_0(y,\xi_0,\eta)\quad\text{at $h(y,\xi_0,\eta)=0$.}
\end{equation}
We calculate
\begin{align*}
\{B,Z_0\} &=b^{-1}\{bB,Z_0\}+B\{b^{-1}I,Z_0\}
           =b^{-1}\{bB,Z_0\}-b^{-2}B\{bI,Z_0\} \\
          &=b^{-1}\{bB,Z_0\}+b^{-2}\{bI,B\}Z_0+b^{-2}\{h,b\}I.
\end{align*}
Here we have taken the Poisson bracket of \eqref{rpt-factor-B-of-Z} with $bI$ 
to obtain  the last equality.
Hence
\[\{B,Z_0\}W=b^{-1}\{bB,Z_0\}W+b^{-2}\{h,b\}W\quad\text{if $W(y)\in N(y)$.}\]
It is well-known that $\{h,f\}$ vanishes on the set of zeros of $h$, if $f$ does.
Therefore, in view of \eqref{b-eq-detprime},
the scalar Poisson bracket $\{h,b\}$ remains unchanged at the zero set of $h$,
if we replace $b$ by $\partial_{\xi_0}\det Z_0$.

By \eqref{def-Asub} and \eqref{eq-psido-compos-sub},
the subprincipal symbol of $\hat Z$ is given by
\begin{align*}
Z_{\sub} &= Z_{-1} -\frac{1}{2i}\sum\nolimits _j \partial^2 Z_0/\partial x^j\partial\xi_j, \\
Z_{-1} &= A_0Q_{-1} +\ii T.
\end{align*}
We know already how to evaluate the second order derivatives of $Z_0$.
The equation \eqref{Lyapunov-AQminus} for $A_0Q_{-1}$ is
a linear equation of the form \eqref{QstarX-XQ-eq-Y}.
The terms on the right-hand side of \eqref{Lyapunov-AQminus} can be evaluated.
This is also true for the normal derivative $\partial_3 Q_0$.
Indeed, the factorization \eqref{factor-of-acoustic-tensor-x-depend},
and the Ricatti equation for $Z_0$ hold also for small $x^3>0$,
so that differentiation with respect to $x^3$ is possible.
Furthermore, the right-hand side of \eqref{Lyapunov-AQminus} contains
the terms $B_0$ and $B_1$.
These are defined in \eqref{def-B0-B1} and in the displayed formula
after \eqref{eqTractionInNormalCoord}.
Notice that the formulas for $B_0$, $B_1$, and for $T$,
defined in \eqref{def-T}, contain Christoffel symbols.
The curvature of $S$ enters into the transport equation only through these expressions.
Summarizing, we have explained how to evaluate $H(y)W_0(y)$
in \eqref{thm-transport-eqn-Wzero} if, as we can assume,
$W_0(y)$ belongs to the null-space $N(y)$ of the surface impedance tensor.

The transport equation involves the ray field $V(y)$
which depends on the eikonal $\varphi(y)$.
The zeros of the secular equation are simple,
so they can be numerically computed in an efficient way.
Differentiating
\[0=\det Z_0(y,-c_R(y,\eta),\eta)\]
by the chain rule,
the derivatives of $c_R$ are found as algebraic expressions
in the derivatives of $Z_0$.
The solution of the eikonal equation for $\varphi$ can be reduced to the
solution of ordinary differential equations by Hamilton--Jacobi theory.
This also allows the computation of the ray field $V$,
the divergence $\Div V$, and the rays $y(t)$.

\begin{remark}
In the special case of a homogeneous body which fills a half-space,
the tensors $Q_0$ and $Z_0$ are independent of surface points.
So the Poisson bracket vanishes, and so does the subprincipal symbol,
because the Christoffel symbols are zero too.
The scalar transport equation \eqref{scalar-transport-eqn} reduces to
\[
\frac{\intd}{\intd t}\psi+ \big(\frac 1 2 \Div  V +\ii\IM\langle W|
    \frac{\intd}{\intd t} W\rangle\big)\psi=0.
\]
If $W$ is constant, this is the differential equation for the spreading factor.
The imaginary term corresponds to the fact that a choice of $W(y)$ is unique
only up to a phase factor $e^{\ii\alpha}$, $\alpha(y)$ real.
In general, for an inhomogeneous, anisotropic body with curved boundary,
we have seen how to evaluate the coefficient $\langle W|HW\rangle$
of the scalar transport equation numerically.
It is desirable, however, to also have a good understanding of the coefficients
in \eqref{scalar-transport-eqn}.
The real part, if positive, would lead to a damping factor.
The imaginary part gives rise to the Berry phase first observed by Babich
in the early 1960's; see \zitatau{BabichKiselev04}{Babich and Kiselev} and the references therein.
\end{remark}

\appendix
\section{Spectral factorization of positive definite matrix polynomials}

The purpose of this appendix is to state and prove the spectral factorization theorem
for self-adjoint matrix polynomials
which is fundamental to our approach to Rayleigh wave theory.
Refer to \zitatau{GohbergLancRodman82matpolyn}{Gohberg et al.} for a comprehensive treatment
of matrix factorizations.

Let $H$ be a finite-dimensional complex Hilbert space, $\dim H=n$.
Let $A(s)$ be a quadratic polynomial with values in the space of linear operators on $H$,
\[
A(s) =A_0s^2+(A_1+A_1^*)s+A_2.
\]
(We denote the adjoint by a star.)
A number $s\in\C$ is called an eigenvalue of the polynomial if $A(s)$ is singular.
The spectrum of $A$ is the set $\sigma(A)$ of its eigenvalues.
Assume that, in addition, the polynomial is self-adjoint, i.e.,
$A(\bar s)=A(s)^*$ holds for all $s$, and that $A_0$ is positive definite.
If $s$ is an eigenvalue of $A$ then so is $\overline{s}$.
If $A$ has no real spectrum, then $A(s)$ is positive definite for real $s$.

In the following, we abuse language and often call linear operators matrices
despite the fact that we do not fix a basis of $H$.

The following is a special case of Theorem~11.2 in \zitatau{GohbergLancRodman82matpolyn}{Gohberg et al.}.

\begin{proposition}
\label{prop-factor-quad-pol}
Assume that $A(s)$ has no real eigenvalues.
Denote by $\sigma_+$ and $\sigma_-$ the intersection of the spectrum of $A(s)$
with the upper and the lower complex half-plane, respectively.
There is a unique matrix $Q$ with spectrum contained in the upper half-plane such that
\begin{equation}
\label{fact-pos-sa-factorisation}
A(s)=(s-Q^*)A_0(s -Q)
\end{equation}
holds for all $s$.
If $\gamma_+$ is a closed Jordan contour which contains $\sigma_+$ in
its interior and $\sigma_-$ in its exterior, then
\begin{equation}
\label{fact-pos-sa-integral-rep}
Q\oint_{\gamma_+} A(t)^{-1}\intd t = \oint_{\gamma_+} t A(t)^{-1}\intd t
\end{equation}
holds.
The integrals are non-singular.
\end{proposition}

If the coefficient matrices $A_j$ depend continuously or differentiably
on some parameters, then, in view of \eqref{fact-pos-sa-integral-rep}, so does $Q$.
This follows from the integral formula \eqref{fact-pos-sa-integral-rep}.

We need some preparations before we can give the proof of the propostion.

Polynomial division of $A(s)$ by $s-Q$ gives
$A(s)=(s-Q_-)A_0(s-Q)+F$ with
$Q_-A_0+A_0Q+A_1+A_1^*=0$ and $Q_-A_0Q+F=A_2$.
A simple calculation shows that the remainder $F$ is zero if and only if
\begin{equation}
\label{fact-pos-sa-solvency}
A_0Q^2+(A_1+A_1^*)Q+A_2=0
\end{equation}
holds.
The divisibility criterion \eqref{fact-pos-sa-solvency} is known as the solvency equation.

The Stroh companion matrix of the polynomial $A(s)$ is the
linear operator $N$ on $H^2$ which is given in block form by
\begin{equation}
\label{factor-Stroh-matrix}
N = \begin{bmatrix} -A_0^{-1}A_1&A_0^{-1}\\ -A_2+A_1^*A_0^{-1}A_1& -A_1^*A_0^{-1}\end{bmatrix}.
\end{equation}
A direct calculation proves the identity
\begin{equation}
\label{factor-Stroh-linearization}
\begin{bmatrix}s A_0+A_1^*& I\\ -A_0& 0\end{bmatrix} (s-N)
   = \begin{bmatrix}A(s)& 0\\ -s A_0 -A_1 & I\end{bmatrix},
\end{equation}
where $s=sI$.
Obviously, the spectrum of $A$ equals the spectrum of $N$.
Moreover, $(u,t)^T$ is an eigenvector of $N$ with eigenvalue $s$
if and only if $A(s)u=0$ and $t=(s A_0+A_1)u$ hold.
Set $L = \begin{bmatrix}I&0\end{bmatrix}$ and $ R = \begin{bmatrix} 0& I \end{bmatrix}^T$.
Passing to inverses in \eqref{factor-Stroh-linearization}, we find that
\begin{equation}
\label{factor-Stroh-resolvent}
A(s)^{-1}=L(s-N)^{-1}R
\end{equation}
holds for all complex numbers $s$.
\begin{remark}
We call $N$ the Stroh companion matrix of the polynomial $A(s)$ because,
in elasticity where $n=3$, $N$ equals Stroh's sextic matrix, \zitatau{Stroh62}{Stroh}.
In that setting, and in the main part of the present paper,
$A(s)$ is the acoustic tensor at $\eta+s\nu$,
where $\nu$ is a surface normal and $\eta$ a horizontal propagation direction.
The Stroh matrix is a convenient example of a linearization of $A(s)$,
which means that \eqref{factor-Stroh-resolvent} holds.
The standard companion matrix of $A(s)$ could however also be used for this purpose.
\end{remark}

The holomorphic functional calculus of $N$ assigns to a function $f$,
which is holomorphic in a neighbourhood of the spectrum $\sigma$ of $N$,
the matrix $f(N)=\oint f(s)(s-N)^{-1}\intd s/2\pi \ii$.
The contour of integration must be chosen such that its 
winding numbers around the eigenvalues of $N$ in the support of $f$
are equal to one.
If $\sigma_0$ is a subset of the spectrum of $N$, and if $f_0=1$ and $f_0=0$
in neighbourhoods of $\sigma_0$ and $\sigma\setminus\sigma_0$, respectively,
then $P_0=f_0(N)$ is a projector,
called the Riesz projector associated with $\sigma_0$.

By \eqref{factor-Stroh-resolvent},
if $f$ is holomorphic in a neighbourhood of the spectrum of $N$, then
\begin{equation}
\label{factor-quad-pol:Stroh-resolvent}
\frac{1}{2\pi i}\oint f(t)A(t)^{-1}\intd t = L f(N)R,
\end{equation}
where the contour of integration must be admissable for the functional calculus.

\begin{proof}[Proof of Proposition~\ref{prop-factor-quad-pol}]
Denote by $P_\pm$ the Riesz projector of $N$ associated with $\sigma_\pm$.
Clearly, $P_+P_-=0=P_-P_+$, and $P_++P_-$ is the unit matrix.
The range $Y_\pm$ of $P_\pm$ is an $N$-invariant subspace of $H^2$,
and the direct sum decomposition $H^2=Y_+ \oplus Y_-$ holds. 
Denote by $N_\pm$ the restriction of $N$ to $Y_\pm$.
Clearly, $\sigma(N_\pm)=\sigma_\pm$.

For sufficiently large $\rho>0$,
we denote by $\gamma_\pm$ the Jordan contour which consists of the
segment $[-\rho,\rho]$ and the semicircle $\rho e^{\pm it}$, $0\leq t\leq \pi$.
By \eqref{factor-quad-pol:Stroh-resolvent} we get the identity
\[
\pm L P_\pm R = \frac{1}{2\pi i} \oint_{\gamma_\pm} A(t)^{-1}\intd t
   = \frac{1}{2\pi i} \int_{-\infty}^\infty A(t)^{-1}\intd t.
\]
The last equality follows when letting $\rho\to\infty$.
Since $A(t)^{-1}$ is positive definite for $t$ real, the integral
defines a non-singular matrix.
We infer that $\dim Y_\pm=n$.
Moreover, $LP_+$ has rank $n$, and
there is a unique linear map $L^\sharp:H\to H^2$ which satisfies $L^\sharp LP_+=P_+$.
Clearly, $L^\sharp=P_+ L^\sharp$ has rank $n$.

We prove that the proposition holds with $Q\equiv L NP_+ L^\sharp$.
Observe that $Q^j=L N^jP_+ L^\sharp$.
Using \eqref{factor-quad-pol:Stroh-resolvent}, we find that
\begin{equation}
\label{fact-sa-pol-integrals}
\frac{1}{2\pi i}\oint_{\gamma_+} s^j A(s)^{-1}\intd s
    = L N^j P_+R = Q^j LP_+ R.
\end{equation}
Using Cauchy's integral theorem, we get
\[(A_0 Q^2+(A_1+A_1^*)Q+A_2)LP_+R=0.\]
This implies the solvency condition \eqref{fact-pos-sa-solvency}.
Thus $A(s)=(s-Q_-)A_0(s-Q)$ holds with some matrix $Q_-$.

Observe that $LP_+L^\sharp$ equals unity.
In view of the definitions, this implies that $Q$ is similar to $N_+$.
Thus the spectrum of $Q$ lies in the upper half-plane.
By \eqref{factor-Stroh-linearization},
we have $\det(A_0)\det(s-N)=\det(A(s))$, and therefore
\[\det(s-N_-)\det(s-N_+)=\det(s-Q_-)\det(s-Q).\]
This implies $\det(s-N_-)=\det(s-Q_-)$.
Hence the spectrum of $Q_-$ equals the spectrum $N_-$,
which is contained in the lower half-plane.
This means that the factorization obtained is spectral.

If $A(s)=(s-Q_-)A_0(s -Q)$ holds with the spectra
of $Q$ and $Q_-$ contained in $\sigma_+$ and $\sigma_-$, respectively,
then \eqref{fact-pos-sa-integral-rep} follows by Cauchy's theorem.
The integral on the left-hand side of
\eqref{fact-pos-sa-integral-rep} equals $2\pi i LP_+ R$ which is non-singular.
The uniqueness of the spectral factorization follows.
Since $a$ is self-adjoint, we also have the factorization
$A(s)=(s-Q^*)A_0(s -Q_-^*)$.
Since the spectra of $Q$ and of $Q_-^*$ are contained in the
upper half-plane, this proves that $Q_-=Q^*$.
\end{proof}

\begin{ack}
The author thanks Aleksei Kiselev and Gerardo Mendoza for valuable remarks and suggestions.
\end{ack}

\def\cprime{$'$}

\end{document}